\documentclass[11pt]{article}

\usepackage[leqno]{amsmath}
\usepackage{stmaryrd}
\usepackage{amsthm}
\usepackage{amsfonts}
\usepackage{amssymb}
\usepackage{eucal}
\usepackage{graphicx}
\usepackage{enumerate}
\usepackage{hyperref}
\usepackage{appendix}

\newcommand{\ket}[1]{\mbox{$\left|{#1}\right\rangle$}}

\newcommand{\eps}{\varepsilon}
\newenvironment{reminder}[1]{\smallskip

\noindent {\bf Reminder of #1 }\em}{\smallskip}

\theoremstyle{plain}
  \newtheorem{theorem}{Theorem}

  \newtheorem{lemma}{Lemma}
  \newtheorem{corollary}{Corollary}

\theoremstyle{definition}

\theoremstyle{remark}

\begin{document}
\title{Quantum algorithms for shortest paths problems in structured instances}
\date{}
\author{Aran Nayebi and Virginia Vassilevska Williams\footnote{Comp. Sci. Dept., Stanford University,  anayebi@stanford.edu and virgi@cs.stanford.edu}}

\maketitle


\begin{abstract}We consider the quantum time complexity of the all pairs shortest paths (APSP) problem and some of its variants. The trivial classical algorithm for APSP and most all pairs path problems runs in $O(n^3)$ time, while the trivial algorithm in the quantum setting runs in $\tilde{O}(n^{2.5})$ time, using Grover search. A major open problem in classical algorithms is to obtain a truly subcubic time algorithm for APSP, i.e. an algorithm running in $O(n^{3-\eps})$ time for constant $\eps>0$. 
To approach this problem, many truly subcubic time classical algorithms have been devised for APSP and its variants for structured inputs.
Some examples of such problems are APSP in geometrically weighted graphs, graphs with small integer edge weights or a small number of weights incident to each vertex, and the all pairs earliest arrivals problem. In this paper we revisit these problems in the quantum setting and obtain the first nontrivial (i.e. $O(n^{2.5-\eps})$ time) quantum algorithms for the problems.

\end{abstract}

\section{Introduction}

The all pairs shortest paths problem (APSP) is one of the most fundamental problems in theoretical computer science.
Classical algorithms such as Floyd-Warshall's~\cite{Floyd,warshall} and Dijkstra's~\cite{dijkstra} give rise to $O(n^3)$ running times for the problem in $n$ node graphs. Slight improvements to this cubic running time followed (e.g.~\cite{Fr76,chan}) culminating in the current best 
$n^3/2^{\Omega(\sqrt{\log n})}$ 
algorithm by Williams~\cite{williamsapsp}. A big open problem is whether one can obtain a truly subcubic time algorithm, running in time $O(n^{3-\eps})$ for $\eps>0$.

In some models of computation, truly subcubic bounds exist for APSP. For instance, Fredman~\cite{Fredman75} showed that the decision tree complexity of the problem is $O(n^{2.5})$. Another natural model of computation are quantum algorithms. There, Grover's search~\cite{grover} gives a simple $\tilde{O}(n^{2.5})$ bound for APSP: it is sufficient to compute the so-called distance product $C$ of two $n\times n$ matrices $A$ and $B$, $C[i,j]=\min_k A[i,k]+B[k,j]$, by applying Grover search over the $n$ choices for $k$ for each $i,j$. Nevertheless, there are no known solutions achieving an $O(n^{2.5-\eps})$ time bound for $\eps>0$.
A natural analogue to the truly-subcubic question in classical algorithms is 
\begin{center}Does APSP have a nontrivial (i.e. an $O(n^{2.5-\eps})$ time for $\eps>0$) quantum algorithm?\end{center}

In the classical setting, for special classes of instances there do exist truly subcubic time algorithms for APSP.
These include the case when the edge weights are integers upper bounded by $n^\delta$ for small enough $\delta$~\cite{sz99,zwickbridge}, when each vertex has a small number of distinct weights incident to it~\cite{Yuster09}, when the weights are on the nodes rather than the edges~\cite{chan07}, or when the weights are chosen at random~\cite{apsprand}.

Perhaps the most general case in which truly subcubic time algorithms are known for APSP is when the input is a so called {\em geometrically weighted graph}.
Roughly speaking, a geometrically weighted graph is a directed or undirected graph $G=(V,E)$, together with a weight function 
$w:E\rightarrow \mathbb{R}$ so that $w(u,v)$ depends on the names of $u$ and $v$ in some way. More formally, let $d,c$ be constants. Let the vertices $V$ be points in $\mathbb{R}^d$, and let $\mathcal{W}$ be a set of $c$ piecewise-algebraic functions on $\mathbb{R}^{d}\times \mathbb{R}^d$ with a constant number of pieces and constant degree. Then $G$ is geometrically weighted if each $w(u,v)$ is chosen from $\{\bar{w}(u, v)\mid \bar{w} \in \mathcal{W}\}\cup \{\infty\}$.


Geometrically weighted graphs encompass many special cases. For instance, they include graphs whose vertices are points in the plane and the edge weights are the Euclidean distances between the end points. 
Other special cases include the case of vertex-weighted graphs and graphs with small integer weights.

Chan~\cite{chan07j} obtained a classical $\tilde{O}(n^{3-(3-\omega)/(2\kappa+2)})$ time algorithm for APSP in geometrically weighted graphs, where 
$\omega<2.373$~\cite{v12,legallmult} is the exponent of square matrix multiplication and
$\kappa$ is a constant that depends on the weight function family $\mathcal{W}$; for node-weighted graphs $\kappa=1$, and for Euclidean graphs $\kappa=3$.
Note that for all constants $\kappa\geq 1$, Chan's running time is truly subcubic and as $\kappa$ grows it gets closer to cubic time.

In this paper we first consider the quantum time complexity of APSP in geometric graphs.
We adapt Chan's technique, combining it with quantum ingredients such as quantum minima finding\footnote{For simplicity we will refer to the quantum minima finding algorithm as Grover's algorithm.}~\cite{minima} (which is a variant of Grover search~\cite{grover}) and quantum Dijkstra's~\cite[Theorem 7.1]{qd} to obtain an 
$O(n^{2.5-\eps})$ time algorithm for $\eps>0$ for all geometric graphs with constant value $\kappa$.
In particular we prove the following theorem.

\begin{theorem}\label{th2}
There exists a quantum algorithm that solves the APSP problem for a geometrically weighted graph $G$ with $n$ vertices in time $\tilde{O}\left(n^{2.5 - (2.5-\omega)/(4\kappa+2)}\right)$ with high probability.
\end{theorem}

Notice that the exponent in Theorem~\ref{th2} is always less than $2.5$ and as $\kappa$ grows, it approaches $2.5$.
As a consequence of Theorem~\ref{th2} APSP has a nontrivial quantum algorithm for node-weighted graphs, graphs with small integer weights, graphs embedded in $\mathbb{R}^d$ for constant $d$ with Euclidean weights for each edge, etc.

The approach of our paper can be applied to practically all versions of APSP for which truly subcubic algorithms are known.
For instance, we outline how to obtain an $\tilde{O}(\sqrt{L}n^{2.5-(2.5-\omega)/6})$ time quantum algorithm for APSP when the input graph has at most $L$ distinct edge weights emanating from each vertex (but these weights can be arbitrary reals).
The best classical algorithm for the problem runs in $\tilde{O}(\sqrt{L}n^{(9+\omega)/4})$ time~\cite{Yuster09}\footnote{Yuster actually obtains a slightly better bound using fast rectangular matrix multiplication. In this paper we phrase everything in terms of $\omega$, but we note that most of our bounds can be improved using rectangular matrix multiplication.}.

The approach also works for variants of the shortest paths problem.
We also show, for instance, that the all pairs nondecreasing path problem (APNP) (also called earliest arrivals) has an $\tilde{O}(n^{2.487})$ time quantum algorithm. In APNP, the length of a path is defined as the weight of the last edge, if the consecutive edge weights form a nondecreasing sequence, and $\infty$ otherwise. The problem has applications to train/flight scheduling~\cite{minty58,v08}.
The fastest classical algorithm for the problem runs in $\tilde{O}(n^{(9+\omega)/4}) \le \tilde{O}(n^{2.843})$ time~\cite{v08,duanpettie}.

\paragraph{Overview of our techniques.}
The main idea behind our techniques is to examine the classical algorithms for the problems and replace key parts of the algorithms with a faster quantum counterpart.

All algorithms for APSP use the following two ingredients in one way or another:
\begin{enumerate}
\item Define a matrix product over some algebraic structure and iterate this matrix product $\ell$ times to compute the distances between pairs of vertices that have shortest paths on few ($\leq \ell$ for parameter $\ell$) nodes. 

\item Pick a random sample $S$ of $\tilde{O}(n/\ell)$ nodes that with high probability hits some shortest path for each pair of nodes with a shortest path on many ($\geq \ell$) nodes. Compute the distances $d(s,v)$ and $d(s,v)$ for all $s\in S$ and $v\in V$ using a variant of Dijkstra's algorithm particular to the shortest paths problem at hand.
\end{enumerate}

The final answer is computed by taking the minimum out of all of the above answers.

Ingredient (1) above requires a fast algorithm to compute the special matrix product. Typically, one computes these matrix products by some sort of partitioning of the input, then using fast integer matrix multiplication to compute some of the answers, typically between different partitions, and then using brute force search within each partition.

To obtain efficient quantum algorithms we first obtain a fast matrix product algorithm for the problem at hand by replacing the brute force portion of the classical matrix product algorithms with Grover search. This results in an $O(n^{2.5-\eps})$ time quantum algorithm for $\eps>0$ whenever there was an $O(n^{3-\eps})$ time classical algorithm.

Then, in ingredient 2, we replace the classical Dijkstra's algorithm with a quantum version~\cite{qd} that runs in $\tilde{O}(n^{1.5})$ time. For some of the variants of the all pairs path problems we consider a modification is necessary, but the runtime is the same. In particular, since in a variant of shortest paths, the weight of a path is not necessarily the sum of its edges but it could be some other function of its edge weights, Dijkstra's algorithm may not work as is. For nondecreasing paths, however, it turns out that the quantum Dijkstra algorithm can be properly modified, as we do in Section~\ref{sec:ext}.

These changes already result in a nontrivial quantum runtime. The distances from ingredient (1) are computed in $\tilde{O}(\ell n^{2.5-\eps})$ time.
The distances $d(s,v)$ from ingredient (2) are computed in $\tilde{O}(n^{2.5}/\ell)$ time, and for each $u,v$ one can obtain $\min_{s\in S}d(u,s)+d(s,v)$ in $\tilde{O}(n^{2} \sqrt{n/\ell})$ time using Grover search. To obtain the best running time, one can set $\ell$ so that $n^{2.5}/\sqrt{\ell} = \ell n^{2.5-\eps}$, and the exponent of the runtime is always $<2.5$.

We go one step further to improve the obtained runtimes by noticing that one does not need to do an extra Grover search to compute $\min_{s\in S}d(u,s)+d(s,v)$. One can instead modify the way ingredient (1) is implemented and reduce the final runtime to roughly $n^{2.5}/\ell + \ell n^{2.5-\eps}$.

In the body of the paper we will carry through the above overview in several contexts.

\paragraph{Related work.}
The closest related work is by
Le Gall and Nishimura~\cite{ln} who considered the complexity of some matrix products over semirings. They showed that using Grover search and fast rectangular matrix multiplication~\cite{legallrect,HP98} one can multiply two matrices over the subtropical ($\max,\min$) semiring in quantum $O(n^{2.48})$ time.
Their result also implies that the all pairs bottleneck paths problem can also be solved in quantum $O(n^{2.48})$ time. Unlike the results of~\cite{ln}, our results apply to all pairs path problems whose matrix products are not over semirings. 

Another related line of work is that on quantum output sensitive matrix multiplication: given two $n\times n$ matrices whose product has at most $L$ nonzeroes, compute their product.
Le Gall~\cite{legallisaac} obtained the current best bound of $\tilde{O}(n\sqrt{L}+L\sqrt{n})$.
Although the author never mentions this, this result also implies that the transitive closure of any given graph can be computed in $\tilde{O}(\min\{n^\omega,L\sqrt n\})$ quantum time, where $L\geq n$ is the number of edges in the transitive closure.

Apart from the above mentioned consequences of prior work, our work presents the first study of all pairs path problems in the quantum setting and we are first to exhibit nontrivial, i.e. $O(n^{2.5-\eps})$ time quantum algorithms for all pairs path problems.

\section{Quantum preliminaries}
For notational convenience, let $[n] = \{1,\ldots, n\}$. We assume that a quantum algorithm can access the entry of any input matrix in a random access manner, which is standard (cf. \cite[\S 2]{ln}). More precisely, given an $n\times n$ matrix $A$ (this can easily be generalized to deal with rectangular matrices), we have an oracle $O_{A}$ such that for any $i,j \in [n]$, and $z \in \{0,1\}^{*}$, maps the state $\ket{i}\ket{j}\ket{0}\ket{z}$ to $\ket{i}\ket{j}\ket{A[i,j]}\ket{z}$. As we are interested in time complexity (and not simply query complexity), we count all the computational steps of the quantum algorithms we give and assign unit cost for each call to the oracle $O_A$. Moreover, when we say ``high probability'' with regard to an algorithm outputting a desired result, we take this to mean with probability at least $2/3$. This probability can be boosted in a standard way to $1-1/\textrm{poly}(n)$.

\section{Motivation: Node-weighted APSP}
Our graph $G = (V, E)$, where $|V| = n$ is given as an $n\times n$ adjacency matrix (where $\infty$ is assigned to the entry $(i,j)$ if and only if $(i,j) \not\in E$), and any query to this matrix is assigned unit cost. Using our quantum tools and techniques due to Chan \cite{chan07}, we can already produce an algorithm for node-weighted APSP (where the weights are integers) that runs in time $\tilde{O}\left(n^{\frac{20+\omega}{9}}\right) \le \tilde{O}(n^{2.486})$ (and of course succeeds with high probability). This leads us to consider the more general case of geometrically weighted graphs, which we discuss in the next section (and which will allow us to even improve the run-time of the algorithm given in this section for the node-weighted case).\newline
\indent Our quantum algorithm for node-weighted APSP proceeds as follows. As done previously in APSP algorithms for integer-weighted graphs, we divide the shortest paths into two categories: those having lengths larger than $s$ nodes, for some parameter $s$, and those having lengths less than $s$ nodes. Recall the standard hitting set argument: Suppose we want to compute the distances $d(u,v)$ for all vertices $u,v$ that have a shortest path $P_{uv}$ on at least $s$ nodes. Then if we sample $O(\frac{n}{s}\log n)$ nodes $S$ independently at random we have that with high probability for all such pairs $u,v$  there is a node of $P_{uv}$ in $S$, and hence for all such $u,v$,
\begin{equation}\label{1}
d(u,v)=\min_{s\in S} (d(u,s)+d(s,v)).
\end{equation}
\indent Given $S$, we first use quantum Dijkstra's to find $d(s,v)$ and $d(v,s)$ for all $v\in V$ and all $s\in S$ in time $\tilde{O}\left(\frac{n}{s} \cdot n^{1.5}\right) = \tilde{O}\left(n^{2.5}/s\right)$. Next, for each $u,v$ as above, we run Grover's algorithm over $S$ to find $d(u,v)$ (using the relation in \eqref{1}). For all couples $(u,v)$ this takes a total time of $\tilde{O}((\sqrt{n/s}) \cdot n^2) = \tilde{O}(n^{2.5}/\sqrt{s})$.
We improve this step in the next section.
\newline
\indent Now, we also compute $d(u,v)$ for all pairs $u,v$ that have shortest paths on $< s$ nodes. This step reduces to repeating the matrix product of a Boolean matrix $A$ with an integer matrix $B$, $s$ times, where $A$ and $B$ are the unweighted adjacency matrix and the matrix recording the shortest distances so far, respectively. Let $C$ be the product of $A$ and $B$, defined as follows:
\begin{equation*}
C[i,j] = \min\{B[k,j]\mid A[i,k] = 1\}.
\end{equation*}
In order to compute $C$, we sort each row of $B$ in increasing order, and let $d$ be a parameter. For each sorted row $k$ of $B$ we partition it into $n/d$ buckets. Next, for all $r \le n/d$, we create
\begin{equation*}
B_r[i,j] := \begin{cases}1\textrm{\mbox{   }if $B[i,j]$ is in bucket $r$ of the $i$-th row of $B$}\\0\textrm{\mbox{   }otherwise.}\end{cases}
\end{equation*}
For all $r$, we compute the standard matrix product $A\cdot B_r$, which takes time $\tilde{O}((n/d)n^{\omega})$. Note that
\begin{equation*}
(A\cdot B_r)[i, j] = 1 \Leftrightarrow \textrm{$\exists k\mbox{    }A[i,k] = 1$ and $B[k,j]$ is in the $r$-th bucket of row $k$ of $B$}.
\end{equation*}
For all $i, j$ let $r_{ij}$ be the minimum $r$ such that
\begin{equation*}
(A\cdot B_r)[i,j] = 1.
\end{equation*}
If $C[i,j] = B[k,j]$ then $B[k,j]$ is in bucket $r_{ij}$ of $B$. Thus, to find $C[i,j]$, we just search bucket $r_{ij}$ using Grover's algorithm. Hence, it takes additional time $\tilde{O}(n^2\sqrt{d})$ to compute $C$. Therefore, the total time to compute $C$ is
\begin{equation*}
\tilde{O}\left(\left(\frac{n}{d}\right)n^{\omega} + n^2\sqrt{d}\right).
\end{equation*}
Taking $d = n^{\frac{2\omega - 2}{3}}$, we get a run-time of $\tilde{O}\left(n^{\frac{5+\omega}{3}}\right)$ for computing $C$.\newline
\indent Thus, the total run-time for node-weighted APSP becomes
\begin{equation*}
\tilde{O}(s\cdot n^{\frac{5+\omega}{3}} + n^{2.5}/\sqrt{s}).
\end{equation*}
We set
$s^{3/2}=n^{2.5-\left(\frac{5+\omega}{3}\right)}$, the runtime is $\tilde{O}\left(n^{\frac{20+\omega}{9}}\right)$.

\section{Geometrically weighted graphs}
Geometrically weighted graphs are discussed by Chan \cite[\S 3]{chan}. It is easy to see that our consideration of node-weighted graphs in the previous section is a special case of a geometrically weighted graph. Here, each vertex is a point in $\mathbb{R}^1$ and $w(p,q) = p$.\newline
\indent Let $A\star B$ be the matrix $C = \{c_{ij}\}_{(i,j)\in [n]\times [n]}$ with $c_{ij} = \min_k(a_{ik} + b_{kj})$. 
$A\star B$ is the {\em distance product} of $A$ and $B$.

Let $A\wedge B$ be the matrix $C = \{c_{ij}\}_{(i,j)\in [n]\times [n]}$ with $c_{ij} = \min\{a_{ij}, b_{ij}\}$. Let $\delta_{G}(p_i, p_j)$ denote the shortest path distance from $i$ to $j$. Note that in the sections that follow we will just describe how to compute distances (one can easily modify this to generate the shortest paths). As Chan does, we will ignore issues about sums of square roots in the case of Euclidean distances.\newline
\indent Note that while we still use the same strategy of dividing shortest paths into two categories, we cannot handle shortest paths of lengths smaller than $s$ by repeated squaring (as many other algorithms do), since the square of a geometrically weighted matrix is no longer geometrically weighted. As a result, we rely on Chan's strategy to compute the \emph{distance product} of a geometrically weighted matrix with an arbitrary matrix, and find paths of small lengths by computing this product $s$ times. 

The key to the truly subcubic classical algorithm for APSP in geometrically weighted graphs is exactly the truly subcubic algorithm for multiplying a geometrically weighted matrix with an arbitrary matrix. Similarly, the key to a nontrivial quantum algorithm for the problem is a nontrivial quantum algorithm for the corresponding matrix product.
For general distance product there are only two quantum algorithms known. The first is the trivial $\tilde{O}(n^{2.5})$ algorithm based on Grover search, and the second is an
$O\left(2^{0.640\ell}n^{\frac{5+\omega}{3}}\right)\le O\left(2^{0.640\ell}n^{2.458}\right)$ time algorithm for computing the $\ell$ most significant bits of the distance product of two matrices with entries in $\mathbb{Z} \cup \{\infty\}$ by Le Gall and Nishimura \cite[Theorem 4.2]{ln}, based on the classical algorithm of Vassilevska and Williams~\cite{vw06}. 

\subsection{Preliminaries}
\indent For ease of reading, we retain much of the same notation as Chan \cite{chan07} uses. We rely on the partition theorem mentioned by Chan \cite[Lemma 3.1]{chan} (due to Agarwal and Matou\v{s}ek~\cite{AgM94}):
\begin{lemma}[Partition theorem]\label{l1}
Let $P$ be a set of $n$ points in $\mathbb{R}^{d+1}$, and let $1 \le r \le n$ be a parameter. Then there is a constant $\kappa$, such that the following is true: We can partition $P$ into $r$ subsets $P_1,\ldots, P_r$, each of size $O(n/r)$, and find $r$ cells $\Delta_1 \supset P_1,\ldots, \Delta_r \supset P_r$, each of complexity $O(1)$, such that any surface of the form
\begin{equation*}
\{(x,z) \in \mathbb{R}^{d+1}\mid w(p,x) + z = c_0\}\mbox{    }(w \in \mathcal{W}, p \in \mathbb{R}^{d}, c_0 \in \mathbb{R})
\end{equation*}
intersects at most $\tilde{O}(r^{1-1/\kappa})$ cells. The subsets and cells can be constructed in $\tilde{O}(n)$ time.
\end{lemma}
$\kappa$ depends on the particular family $\mathcal{W}$ of weight functions, and the best possible value of $\kappa$ is an open problem for dimensions greater than 4. However, for the case of node-weighted graphs (where $w(p,q) = p$), we can take $\kappa = 1$, and for Euclidean graphs (where $w$ is the Euclidean distance), we can take $\kappa = 3$.

\begin{theorem}\label{th1}
There exists a quantum algorithm such that given a geometrically weighted $n\times n$ matrix $A$ and an arbitrary $n\times n$ matrix $B$, outputs $C = A\star B$ in $\tilde{O}\left(n^{\frac{5\kappa + \omega}{2\kappa + 1}}\right)$ time, with high probability.
\end{theorem}

\begin{proof}
Our algorithm follows similar lines as the one given by Theorem 3.2 of Chan \cite{chan07}. For simplicity, we assume $\mathcal{W}$ consists of a single function $w$ (if $\mathcal{W}$ contains $c$ functions, we can just run the algorithm we provide below $c$ times and return the element-wise minimum). Let $p_1,\ldots, p_n \in \mathbb{R}^d$ denote the points that define $A$. For each $j \in [n]$, we apply Lemma~\ref{l1} to the set of points $\{(p_k, b_{kj})\}_{k\in [n]}$ to obtain $r$ subsets $\{P_{\ell j}\}_{\ell = 1}^{r}$, each of size $O(n/r)$ and $r$ cells $\{\Delta_{\ell j}\}_{\ell = 1}^{r}$. This step takes total time $\tilde{O}(n^2)$.\newline
\indent Let $F_i$ be the set of \emph{indices} of finite entries in $A$, namely, $F_i = \{k\mid a_{ik} = w(p_i, p_k)\}$. Note that, by definition, $c_{ij} = \min_{k \in F_i}(w(p_i, p_k) + b_{kj})$, for each $i$ and $j$. Now, to compute $c_{ij}$, for every $i, j \in [n]$, $\ell \in [r]$, we first determine whether $F_i$ intersects $P_{\ell j}$. Here, we abuse notation and let $P_{\ell j}$ contain \emph{indices} $k$ instead of the points $(p_k, b_{kj})$. Let $D$ be an $n\times n$ Boolean matrix whose rows correspond to bit vectors of the $F_i$'s, and let $E$ be an $n\times nr$ Boolean matrix whose columns correspond to bit vectors of the $P_{\ell j}$'s. Then this preprocessing step reduces to multiplying $D$ by $E$, which takes $O(n^{\omega}r)$ time.\newline
\indent Now, let ${\hat{c}}_{ij} = \underset{\ell: P_{\ell j}\cap S_i \ne \emptyset}{\min}\sup_{(x,z)\in\Delta_{\ell j}}(w(p_i, x)+z)$. Note that ${\hat{c}}_{ij}$ is an upper bound on the actual value of $c_{ij}$ and can be computed in $\tilde{O}(\sqrt{r})$ time by Grover's algorithm ($O(1)$ time per $\ell$). Let $\gamma_{ij}$ be the region $\{(x,z) \in \mathbb{R}^{d+1}\mid w(p_i, x) + z \le {\hat{c}}_{ij}\}$. Set $c_{ij}$ to be the minimum of $w(p_i, p_k) + b_{kj}$ over all $k \in P_{\ell j}\cap S_i$ and over all $\ell$ with $\Delta_{\ell j}$ intersecting the boundary of $\gamma_{ij}$. By Lemma~\ref{l1}, the number of such $\ell$'s is $\tilde{O}(r^{1-1/\kappa})$ and each $P_{\ell j}$ has $O(n/r)$ size, so using Grover's algorithm to compute $c_{ij}$ this step takes time $\tilde{O}(\sqrt{r^{1-1/\kappa}\cdot n/r}) = \tilde{O}(\sqrt{n}/r^{1/2\kappa})$.\newline
Therefore, the total running time is
\begin{equation*}
\tilde{O}\left(n^{\omega}r + n^{2.5}/r^{1/2\kappa}\right).
\end{equation*}
Taking $r = n^{(2.5-\omega)\frac{2\kappa}{2\kappa + 1}}$, we get the desired run-time.
\newline
\indent We omit the proof of correctness as it can be found in the proof of Theorem 3.2 of Chan \cite{chan07}.
\end{proof}
Observe that as $\kappa$ increases, the run-time of the algorithm given in Theorem~\ref{th1} approaches $\tilde{O}(n^{2.5})$. In particular, for the cases $\kappa = 1$ and $\kappa = 3$, the run-times are $\tilde{O}\left(n^{\frac{5+ \omega}{3}}\right) \le \tilde{O}(n^{2.458})$ and $\tilde{O}\left(n^{\frac{15 + \omega}{7}}\right) \le \tilde{O}(n^{2.482})$, respectively.\newline
\indent Also, note that for the case $\kappa = 1$, Theorem~\ref{th1} gives us a $\tilde{O}\left(n^{\frac{5+ \omega}{3}}\right)$ time quantum algorithm for computing the matrix product $C$ of a Boolean matrix $A$ with a real valued matrix $B$, defined as
\begin{equation*}
C[i,j] = \min\{B[k,j]\mid A[i,k] = 1\}.
\end{equation*}
We will use this matrix product in Theorem~\ref{th7}.

The best classical runtime for the above matrix product is the same as that for max-min product~\cite{duanpettie}. However, in the quantum setting we obtain a {\em better} algorithm for the above product than the best known quantum algorithm for max-min product by Le Gall and Nishimura~\cite{ln} that runs in $O(n^{2.473})$ time.

\subsection{APSP algorithm}

\begin{reminder}{Theorem~\ref{th2}}
There exists a quantum algorithm that solves the APSP problem for a geometrically weighted graph $G$ with $n$ vertices in time $\tilde{O}\left(n^{\frac{10\kappa + \omega + 2.5}{2(2\kappa + 1)}}\right)$ with high probability.
\end{reminder}

\begin{proof}
Our algorithm follows similar lines as the one given by Theorem 3.4 of Chan \cite{chan07}. Let $A$ be the corresponding weight matrix of $G$. By iterating the distance product of $A$ with itself $2 \le s \le \ell$ times, we can compute for each $i$ and $j$ the weight of the shortest length-$s$ path from $i$ to $j$. In other words, taking $A = A^{(1)}$, for each $s = 2,\ldots, \ell$, we compute $A^{(s)} = A\star A^{(s-1)}$. The entry $a^{(s)}_{ij}$ is the weight of the shortest length $s$ path from $i$ to $j$. This step requires $\ell - 1$ applications of Theorem~\ref{th1}, and hence takes time $\tilde{O}\left(\ell n^{\frac{5\kappa + \omega}{2\kappa + 1}}\right)$.
Let $\bar{A} = A^{(1)}\wedge\ldots\wedge A^{(\ell)}$. $\bar{A}$ contains all distances between pairs of nodes with shortest paths on $\leq s$ nodes.

\indent By a hitting set argument, we can find a subset $S$ of size $\tilde{O}(n/\ell)$ that hits all $O(n^2)$ shortest length-$\ell$ paths found in $O(n^2\ell)$ time. Let $B^{(0)} = B = \{b_{ij}\}_{(i,j)\in [n]\times [n]}$, such that
\begin{equation*}
b_{ij} := \begin{cases}
\delta_G(p_i, p_j)\mbox{    }\textrm{if $p_i \in S$,}\\
\infty\mbox{     }\textrm{if otherwise.}
\end{cases}
\end{equation*}
We can construct $B^{(0)}$ by $\tilde{O}(n/\ell)$ applications of quantum Dijkstra's, giving us a total time of $\tilde{O}(n^{2.5}/\ell)$ for this step.\newline
\indent For each $s = 1,\ldots, \ell$, compute $B^{(s)} = A\star B^{(s-1)}$. This requires $\ell$ applications of Theorem ~\ref{th1}, and therefore this step also takes time $\tilde{O}\left(\ell n^{\frac{5\kappa + \omega}{2\kappa + 1}}\right)$.\newline

Let $\bar{B}=B^{(0)}\wedge\ldots\wedge B^{(\ell)}$. Notice that for every $u,v$ with some shortest path on $\geq s$ nodes, $\delta_G(u, v)=\min_{x\in S} \delta_G(u,x)+\delta_G(x,v)$ since after iterating the geometric weight distance product $s$ times, some node of $S$ that intersects the shortest $u,v$ path is hit and hence with the next geometric product, the correct distance is computed. This is exactly where we have avoided doing a Grover search over $S$ to find  $\min_{x\in S} \delta_G(u,x)+\delta_G(x,v)$, replacing it with another series of matrix products. Hence the runtime is improved over the runtime obtained in the Motivation section for APSP in node weighted graphs.

\indent Finally, we return $\bar{A}\wedge \bar{B}$.
This will give us the shortest path from $p_i$ to $p_j$ (correctness follows from the proof of Theorem 3.4 of Chan \cite{chan07}).\newline
\indent The total running time is
\begin{equation*}
\tilde{O}\left(\ell n^{\frac{5\kappa + \omega}{2\kappa + 1}} + n^{2.5}/\ell\right).
\end{equation*}
Taking $\ell = n^{\frac{2.5 - \omega}{4\kappa + 2}}$, gives us our desired run-time.
\end{proof}

Note that as $\kappa$ increases, the run-time of the algorithm given in Theorem~\ref{th2} approaches $\tilde{O}(n^{2.5})$. In particular, for the cases $\kappa = 1$ and $\kappa = 3$, the run-times are $\tilde{O}\left(n^{\frac{12.5+ \omega}{6}}\right) \le {O}(n^{2.479})$ and $\tilde{O}\left(n^{\frac{32.5 + \omega}{14}}\right) \le {O}(n^{2.491})$, respectively.

\section{Extensions}\label{sec:ext}
Here we apply our framework to more contexts.
%
%
The proofs of some our results in this section will be restricted to sketches, and full proofs of Theorems~\ref{th3} and~\ref{th4} can be found in the appendix.\newline
\paragraph{More on geometrically weighted graphs.}
\indent First, we consider geometrically weighted graphs. We can improve our result in Theorem~\ref{th2} for the related problems of APSP where the weights are between 1 and $c$.
As Chan \cite{chan07} does, we first consider a variant of Theorem~\ref{th1}, which considers the sparseness of the input matrices:

\begin{theorem}\label{th3}
There exists a quantum algorithm such that given a geometrically weighted $n\times n$ matrix $A$ and an arbitrary $n\times n$ matrix $B$, where $B$ has $O(m)$ finite entries, outputs any $O(m)$ specified entries of $C = A\star B$ in $\tilde{O}\left(n^{\omega} + mn^{\frac{\kappa+\omega - 2}{2\kappa + 1}}\right)$ time with high probability.
\end{theorem}

\begin{proof}
This is a simple modification of the proof of Theorem~\ref{th1}. Our algorithm follows similar lines as the one given by Theorem 3.5 of Chan \cite{chan07}. The preprocessing step takes time $O\left(n^{\omega} + rmn^{\omega - 2}\right)$, as we are multiplying an $n\times n$ Boolean matrix with an $n\times rm/n$ matrix. As before in Theorem~\ref{th1}, each $c_{ij}$ can be computed in time $\tilde{O}(\sqrt{r} + \sqrt{n}/r^{1/2\kappa})$.\newline
\indent Therefore, the total running time is
\begin{equation*}
\tilde{O}\left(n^{\omega} + rmn^{\omega - 2} + m\sqrt{n}/r^{1/2\kappa}\right).
\end{equation*}
Taking $r = n^{(2.5 - \omega)\frac{2\kappa}{2\kappa + 1}}$, gives us the desired run-time.
\end{proof}

We now apply Theorem~\ref{th3} as follows:

\begin{theorem}\label{th4}
There exists a quantum algorithm such that given a geometrically weighted graph $G$ with $n$ vertices, can solve the APSP problem in time $\tilde{O}\left(n^{\frac{5\kappa+\omega}{2\kappa+1}}\right)$ with high probability, assuming that the weights are between 1 and $c$.
\end{theorem}

\begin{proof}
Our algorithm follows similar lines as the one given by Theorem 3.6 of Chan \cite{chan07}. For each $s = 1,\ldots,c\ell$, we will first compute the matrix $A^{(s)} = \{a^{(s)}_{ij}\}_{(i,j)\in [n]\times [n]}$, as Chan \cite{chan07} does in his proof of Theorem 3.6. However, we use Theorem~\ref{th3} instead of his Theorem 3.2 and get a run-time of $\tilde{O}\left(\ell n^{\omega} + n^{\frac{5\kappa+\omega}{2\kappa+1}}\right)$ in order to construct $A^{(s)}$ for all $s = 1,\ldots, c\ell$.\newline
\indent Next, by a hitting set argument, we find a subset $S$ of size $\tilde{O}(n/\ell)$ that hits all shortest paths of length exactly $\ell$, which takes $O(n^2\ell)$ time. We can compute the matrices $B$ and $B'$ by $\tilde{O}(n/\ell)$ applications of quantum Dijkstra's. Therefore, the total time for this step is $\tilde{O}(n^{2.5}/\ell)$.\newline
\indent Finally, we compute $B'\star B$ in $\tilde{O}(n^2\sqrt{|S|}) = \tilde{O}(n^{2.5}/\sqrt{\ell})$ time using the trivial quantum algorithm for computing distance product. Then we return $A^{(1)}\wedge\cdots\wedge A^{(c\ell)}\wedge (B'\star B)$. We omit the proof of correctness as it can be found in Theorem 3.6 of Chan \cite{chan07}.\newline
\indent Thus, the total run-time is
\begin{equation*}
\tilde{O}\left(\ell n^{\omega} + n^{\frac{5\kappa+\omega}{2\kappa+1}} + n^{2.5}/\sqrt{\ell}\right).
\end{equation*}
Taking $\ell = n^{\frac{5-2\omega}{3}}$ gives us our desired run-time.
\end{proof}

\paragraph{APSP in graphs with a small number of weights incident to each vertex.}

\begin{theorem}\label{th7}
There exists a quantum algorithm such that given a directed graph $G$ with $n$ vertices with at most $L$ distinct edge weights emanating from each vertex, can solve the APSP problem in time $\tilde{O}(\sqrt{L}n^{2.5-(2.5-\omega)/6})$ with high probability.
\end{theorem}
\begin{proof}
We divide the task into computing short paths (over $\le s$ nodes) with successive matrix multiplications and computing long paths with quantum Dijkstra's from $\tilde{O}(n/s)$ nodes via a hitting set argument. For long paths, this takes total time $\tilde{O}(n^{2.5}/s)$.\newline
\indent For short paths, we let $B_z[i,j]$ denote the distance from $i$ to $j$ over paths of length $z$. Let $A_{\ell}$ be the Boolean matrix where $A_{\ell}[u,v] = 1$ if and only if $(u,v)$ is an edge in $G$ of the $\ell$-th largest weight from $u$, where $\ell \in [L]$. Let $\ell[u]$ be the value of the $\ell$-th largest weight from $u$. Thus,
\begin{equation*}
B_{z+1}[i,j] = \min_{\ell\in [L]}\{\ell[i] + \min_{k}\{B[k,j]\mid A_{\ell}[i,k] = 1\}\}
\end{equation*}
With $B_1 = A$ and using our algorithm from Theorem ~\ref{th1}, we compute this product $s$ times, giving us a total time of $\tilde{O}\left(Lsn^{\frac{5+\omega}{3}}\right)$ for this step.\newline
\indent Thus, the total running time is
\begin{equation*}
\tilde{O}\left(n^{2.5}/s + Lsn^{\frac{5+\omega}{3}}\right).
\end{equation*}
Setting $s = \left(L^{-1}n^{2.5-\frac{5+\omega}{3}}\right)^{1/2}$, we get our desired run-time.
\end{proof}

\paragraph{APNP: an example of a variant of APSP.}
Here we consider the problem of all pairs nondecreasing paths (APNP), also known as earliest arrivals. Here, we want to find, for every pair of vertices $i$ and $j$, a path with edge weights of nondecreasing length such that the last edge weight on the path is minimized \cite{v08,duanpettie}. We have the following result:

\begin{theorem}\label{th8}
There exists a quantum algorithm such that given a directed graph $G$ with $n$ vertices, we can solve the APNP problem in time $\tilde{O}(n^{2.487})$ with high probability.
\end{theorem}
\begin{proof}
We divide the task into computing short paths (of length $\le s$) with successive matrix multiplications and computing long paths with a variant of quantum Dijkstra's from $\tilde{O}(n/s)$ nodes via a hitting set argument.\newline
\indent For short paths, we let $A$ be the matrix such that $A[i,i] = -\infty$, and if $i \ne j$, $A[i,j] = w(i,j)$ if there is a directed edge from $i$ to $j$, where $w$ is the weight of this edge, and $A[i,j] = \infty$, if otherwise. First set $B = A$ and then repeat $B = B \olessthan A$, $s-1$ times, where $\olessthan$ is the $(\min, \le)$ product. Using the $O(n^{2.473})$ time quantum algorithm of Le Gall and Nishimura \cite[Theorem 4.1]{ln} to compute the $(\min, \le)$ product, the running time of this step is $O(sn^{2.473})$. This gives for each $i, j$ the minimum weight of a last edge on a nondecreasing path from $i$ to $j$ of length at most $s$.\newline

\indent When running the above, for each pair of vertices $i, j$, keep an actual minimum nondecreasing path from $i$ to $j$ of length at most $s$. A subset $K$ of these paths are of length exactly $s$. Let $S$ be our hitting set of size $\tilde{O}(n/s)$, which hits each of the paths in $K$.

For each $v \in S$, we run a modified version of quantum Dijkstra's \cite[Theorem 7.1]{qd}, to compute for every $u\in V$ the minimum last edge weight on a nondecreasing path from $u$ to $v$. Let this value be $W(u,v)$ if a nondecreasing path exists and $\infty$ otherwise.
In the classical setting~\cite{v08j}, one iterates through all the neighbors $z$ of $v$ in nondecreasing order of their weight.
For each such neighbor $z$, one
performs a Dijkstra search from $z$ in the graph with edge directions reversed, relaxing an edge $(x,y)$ iff the weight of the original edge $(y,x)$ is at most the current computed distance at $x$, setting the distance of $y$ to $w(y,x)$, setting the key of $y$ to $-w(y,x)$, and labeling $y$ with $w(v,z)$. 
Every time a node's distance is computed, it is removed from the entire graph and never accessed again. If a node is touched by a relaxed edge in a call for some $z$ then its distance is eventually computed in that call and it is removed.

To obtain a quantum version of the above algorithm, we replace each call to Dijkstra's above with a quantum Dijkstra's algorithm.
 The quantum version of Dijkstra's algorithm in \cite{qd} does not use the fact that the weight of a path is the sum of the edge weights, so that this classical variant of Dijkstra's can be implemented. 

Within polylog factors, the cost is $\sum_{z\in N(v)} n_z^{1.5}$, where $n_z$ is the number of nodes touched in the call to node $z$. Thus within polylog factors, the time complexity of the variant of quantum Dijkstra's we need is $\sum_{z\in N(v)} n_z^{1.5}\leq \sqrt{n}\sum_{z\in N(v)} n_z = n^{1.5}$.


\indent Therefore, the total running time of the APNP algorithm is
\begin{equation*}
\tilde{O}(sn^{2.473} + n^{2.5}/s),
\end{equation*}
and taking $s = n^{0.0135}$ gives us the desired run-time.
\end{proof}

\paragraph{Minimum triangle and APSP.}
The minimum weight triangle problem is, given an edge weighted graph $G$, find a triangle of minimum weight sum.
An important result is that classically the minimum weight triangle and APSP are {\em subcubically equivalent}~\cite{focs10}.
There is a trivial $\tilde{O}(n^{1.5})$
time quantum algorithm to compute the minimum weight triangle in a graph.
Here we show as a simple corollary an analogue of the equivalence theorem of~\cite{focs10} in the quantum setting.
\begin{corollary}
If for some constant $\eps>0$ there exists an $O(n^{1.5-\eps})$ time quantum algorithm for minimum triangle, then there is an $O(n^{2.5-\delta})$ time algorithm for distance product and hence APSP for some constant $\delta>0$.
\end{corollary}

\begin{proof}
The proof in \cite{focs10} reduces APSP to $\tilde{O}(n^2)$ instances of the minimum weight triangle problem on $n^{1/3}$ node graphs.
Thus, if minimum weight triangle can be solved in $O(n^{1.5-\eps})$ quantum time, then applying the above classical reduction we obtain that APSP is in quantum $\tilde{O}(n^{2.5-\eps/3})$ time.
%
%
\end{proof}
Thus, to improve on the quantum runtime for APSP, one should concentrate on finding a nontrivial quantum algorithm for minimum weight triangle.
There are no known lower bounds for the problem. The only related result is an $\tilde{\Omega}(n^{9/7})$ lower bound on the query complexity of the harder problem of finding a triangle of weight sum exactly $0$~\cite{belovs}.

%
%

\newpage
\bibliographystyle{amsplain}
\bibliography{references}

\newpage
\appendix
\section{Appendix}
\subsection{Proof of Theorem~\ref{th3}}
This is a simple modification of the proof of Theorem~\ref{th1}. Our algorithm follows similar lines as the one given by Theorem 3.5 of Chan \cite{chan07}. For each $j$, we let $G_j$ be the set of \emph{indices} of finite entries in the $j$-th column of $B$, namely, $G_j = \{k\mid b_{kj} \ne \infty\}$. We now apply Lemma~\ref{l1} to the set of points $\{(p_{k}, b_{kj})\}_{k \in T_j}$. Now, we still maintain that each subset $P_{\ell j}$ has size $O(n/r)$, but the number of subsets is instead $r_j := O(r|T_j|/n)$. Since the total number of subsets is now $\sum_j r_j = O(rm/n)$, then the preprocessing step takes time $O\left(n^{\omega} + rmn^{\omega - 2}\right)$, as we are multiplying an $n\times n$ Boolean matrix with an $n\times rm/n$ matrix. As before, each $c_{ij}$ can be computed in time $\tilde{O}\left(\sqrt{r_j} + \sqrt{{r_j}^{1-1/\kappa}\cdot n/r}\right) \le \tilde{O}(\sqrt{r} + \sqrt{n}/r^{1/2\kappa})$.\newline
\indent Therefore, the total running time is
\begin{equation*}
\tilde{O}\left(n^{\omega} + rmn^{\omega - 2} + m\sqrt{n}/r^{1/2\kappa}\right).
\end{equation*}
Taking $r = n^{(2.5 - \omega)\frac{2\kappa}{2\kappa + 1}}$, gives us the desired run-time.

\subsection{Proof of Theorem~\ref{th4}}
Our algorithm follows similar lines as the one given by Theorem 3.6 of Chan \cite{chan07}. Let $A = \{a_{ij}\}_{(i,j)\in [n]\times [n]}$ be the corresponding weight matrix of $G$. For each $s = 1,\ldots,c\ell$, we will first compute the matrix $A^{(s)} = \{a^{(s)}_{ij}\}_{(i,j)\in [n]\times [n]}$, where
\begin{equation*}
a^{(s)}_{ij}:=\begin{cases}
\delta_G(i,j)\mbox{    }\textrm{if $\delta_G(i,j)\in [s, s+1)$,}\\
\infty\mbox{    }\textrm{if otherwise.}
\end{cases}
\end{equation*}
We compute $A^{(s)}$ as Chan \cite{chan07} does in his proof of Theorem 3.6, however, we use Theorem~\ref{th3} instead of his Theorem 3.2 and get a run-time of $\tilde{O}\left(\ell n^{\omega} + n^{\frac{5\kappa+\omega}{2\kappa+1}}\right)$ in order to construct $A^{(s)}$ for all $s = 1,\ldots, c\ell$.\newline
\indent Next, by a hitting set argument, we find a subset $S$ of size $\tilde{O}(n/\ell)$ that hits all shortest paths of length exactly $\ell$, which takes $O(n^2\ell)$ time. Let $B = \{b_{ij}\}_{(i,j)\in [n]\times [n]}$ be such that
\begin{equation*}
b_{ij} := \begin{cases}
\delta_G(p_i, p_j)\mbox{    }\textrm{if $p_i \in S$,}\\
\infty\mbox{     }\textrm{if otherwise.}
\end{cases}
\end{equation*}
Next, let $B' = \{{b'}_{ij}\}_{(i,j)\in [n]\times [n]}$ be such that
\begin{equation*}
{b'}_{ij} := \begin{cases}
\delta_G(p_i, p_j)\mbox{    }\textrm{if $p_j \in S$,}\\
\infty\mbox{     }\textrm{if otherwise.}
\end{cases}
\end{equation*}
$B$ and $B'$ can be computed by $\tilde{O}(n/\ell)$ applications of quantum Dijkstra's. Therefore, the total time for this step is $\tilde{O}(n^{2.5}/\ell)$.\newline
\indent Finally, we compute $B'\star B$ in $\tilde{O}(n^2\sqrt{|S|}) = \tilde{O}(n^{2.5}/\sqrt{\ell})$ time using the trivial quantum algorithm for computing distance product. Then we return $A^{(1)}\wedge\cdots\wedge A^{(c\ell)}\wedge (B'\star B)$. We omit the proof of correctness as it can be found in Theorem 3.6 of Chan \cite{chan07}.\newline
\indent Thus, the total run-time is
\begin{equation*}
\tilde{O}\left(\ell n^{\omega} + n^{\frac{5\kappa+\omega}{2\kappa+1}} + n^{2.5}/\sqrt{\ell}\right).
\end{equation*}
Taking $\ell = n^{\frac{5-2\omega}{3}}$ gives us our desired run-time.

\end{document}